\documentclass[11pt]{article}
\usepackage{latexsym}
\usepackage{a4wide}
\usepackage{amssymb,amsmath}
\usepackage{amsthm}
\usepackage{amsfonts}

\newtheorem{theorem}{Theorem}
\newtheorem{lemma}[theorem]{Lemma}

\newtheorem{claim}{Claim}

\newtheorem{corollary}{Corollary}[theorem]
\theoremstyle{definition}

\theoremstyle{remark}

\newcommand{\diff}{{\bf Diff}}
\newcommand{\glt}{{\bf GL.3}}
\newcommand{\kfour}{{\bf K4}}
\newcommand{\sfour}{{\bf S4}}
\newcommand{\sfive}{{\bf S5}}
\newcommand{\kt}{{\bf K3}}
\newcommand{\kft}{{\bf K4.3}}
\newcommand{\sft}{{\bf S4.3}}
\newcommand{\kkk}{{\bf K}}
\newcommand{\pskf}{{\bf K4}^-}
\newcommand{\finfone}{\varphi_\infty}
\newcommand{\finfoner}{\varphi_\infty^\bullet}
\newcommand{\finftwo}{\psi_\infty}
\newcommand{\diag}{p}
\newcommand{\next}{q}
\newcommand{\commut}{[L_0,L_1]}
\newcommand{\fusion}{L_0\oplus L_1}
\newcommand{\Log}{\sf Log}

\newcommand{\DhM}{\blacklozenge_0}
\newcommand{\Dh}{\Diamond_0}
\newcommand{\Dv}{\Diamond_1}
\newcommand{\BhM}{\blacksquare_0}
\newcommand{\Bh}{\Box_0}
\newcommand{\Bv}{\Box_1}
\newcommand{\Dhe}{\Dh^{\mbox{\tiny =}{\bf 1}}}
\newcommand{\Dhn}{\Dh^{\mbox{\tiny =}{\bf n}}}
\newcommand{\Dhnn}{\Dh^{\mbox{\tiny =}{\bf n+1}}}

\newcommand{\RhM}{\overline{R}_0^{\M}\!}
\newcommand{\Rh}{R_0}
\newcommand{\Rv}{R_1}

\newcommand{\wconM}{{\rm (wcon$^-$)}$^\M$}
\newcommand{\lcomM}{{\rm (lcom)}$^\M$}
\newcommand{\rcomM}{{\rm (rcom)}$^\M$}
\newcommand{\confM}{{\rm (conf)}$^\M$}

\newcommand{\wcon}{{\rm (wcon$^-$)}}
\newcommand{\lcom}{{\rm (lcom)}}
\newcommand{\rcom}{{\rm (rcom)}}
\newcommand{\conf}{{\rm (conf)}}

\newcommand{\F}{\mathfrak F}

\newcommand{\M}{\mathfrak M}
\newcommand{\Fr}{\sf Fr\,}

\newcommand{\auf}{(}
\newcommand{\zu}{)}
\newcommand{\mprod}{\!\times\!}

\begin{document}

  \title{Bimodal Logics with a `Weakly Connected' Component\\
  without the Finite Model Property}

  \author{Agi Kurucz\\[5pt]
{\small Department of Informatics}\\
    {\small King's College London}}
    
    \date\

\maketitle

\begin{abstract}
  There are two known general results on the finite model property (fmp) of commutators $\commut$ (bimodal logics with commuting
  and confluent modalities). 
  If $L$ is finitely axiomatisable by modal formulas having universal Horn first-order correspondents, then both
$[L,\kkk]$ and $[L,\sfive]$ are determined by classes of frames that
admit filtration, and so have the fmp. On the negative side, 
if both $L_0$ and $L_1$ are determined by transitive frames and have frames of arbitrarily large depth,
then $\commut$ does not have the fmp.
In this paper we show that commutators with a `weakly connected' component often lack the fmp.
Our results imply that the above positive result does not generalise to universally axiomatisable component logics,
and even commutators without `transitive' components such as $[\kt,\kkk]$ can lack the fmp.
We also generalise the above negative result to cases where one of the component logics has frames of
depth one only, such as $[\sft,\sfive]$ and the decidable product logic $\sft\mprod\sfive$. We also show cases when already half of commutativity is enough to force
infinite frames.
\end{abstract}


\section{Introduction}\label{intro}

A normal multimodal logic $L$ is said to have the \emph{finite model property} (\emph{fmp}, for short), if for every
$L$-falsifiable formula $\varphi$, there is a finite model (or equivalently, a finite frame \cite{Segerberg71}) for $L$ where $\varphi$ fails to hold.
The fmp can be a useful tool in proving decidability and/or Kripke completeness of a multimodal logic. While in
general it is undecidable whether a finitely axiomatisable modal logic has the fmp \cite{cz93}, 
there are several general results on the fmp of
unimodal logics (see \cite{cz,WZ07} for surveys and references).
In particular, by Bull's theorem  \cite{Bull66} all extensions of $\sft$ have the fmp. $\sft$ is the finitely axiomatisable modal logic determined by frames $(W,R)$, where $R$ is reflexive, transitive and
\emph{weakly connected}:
\[
\forall x,y,z\in W\,\bigl(xRy\land xR z\to (y=z\lor yRz\lor zRy)\bigr).
\]
The property of
weak connectedness is a consequence of linearity, and so well-studied in temporal and dynamic logics, 
modal-like logical formalisms over point-based models of time and sequential computation \cite{Goldblatt87}.

Here we are interested in to what extent Bull's theorem holds in the bimodal case, that is, we study the fmp of 
bimodal logics with a weakly connected unimodal component. In general,
it is of course much more difficult to understand the behaviour of bimodal logics having two possibly differently behaving modal operators, especially when they interact. Without interaction, there is a general transfer theorem \cite{Fine&Schurz96,Kracht&Wolter91}:
If both $L_0$ and $L_1$ are modal logics having the fmp, then their \emph{fusion} (also known as \emph{independent join}) $\fusion$ also has the fmp.
Here we study bimodal logics with a certain kind of interaction.
Given unimodal logics $L_0$ and $L_1$, their \emph{commutator} $\commut$ is the smallest bimodal logic containing 
their fusion $\fusion$, plus the interaction axioms
\begin{equation}\label{commut}
\Bv\Bh p\to \Bh\Bv p,
\qquad \Bh\Bv p\to \Bv\Bh p,
\qquad \Dh\Bv p\to \Bv\Dh p.
\end{equation}
These bimodal formulas have the respective first-order frame-correspondents of \emph{left commutativity},
\emph{right commutativity}, and \emph{confluence} (or \emph{Church--Rosser property}):

\medskip
\noindent
\lcom\qquad
$\forall x,y,z\ (x\Rh y\Rv z\;\to\;\exists u\ x\Rv u\Rh z),$

\medskip
\noindent
\rcom\qquad
$\forall x,y,z\ (x\Rv y\Rh z\;\to\;\exists u\ x\Rh u\Rv z),$

\medskip
\noindent
\conf\qquad\,
$\forall x,y,z\ \bigl(x\Rv y\;\land\;x\Rh z\;\to\;\exists u\; (y\Rh u\;\land\; z\Rv u)\bigr).$

\medskip
\noindent
These three properties always hold in special two-dimensional structures called \emph{product frames}, and
so commutators always have product frames among their frames. 
Product frames are natural constructions modelling interaction between different domains that might represent
time, space, knowledge, actions, etc. Properties of product frames and \emph{product logics} (logics determined by 
classes of product frames) are extensively studied, see \cite{Gabbay&Shehtman98,gkwz03,Kurucz07} for surveys
and references.
Here we summarise the known results related to the finite model property of commutators and products:

(I) 
It is easy to find bimodal formulas that
`force' infinite ascending or descending chains of points in product frames under very mild assumptions
(see Section~\ref{defs} for details).
Therefore, commutators often do not have the \emph{fmp w.r.t.\ product frames}.
However, commutators and product logics
do have other frames, often ones that are not even p-morphic images of product frames,
or finite frames that are p-morphic images of infinite product frames only (see Section~\ref{defs}).
So in general the lack of fmp of a logic does not
obviously follow from the lack of fmp w.r.t.\ its product frames. In fact, there are known examples, say 
$[\kfour,\kkk]=\kfour\mprod\kkk$ and
$[\sfour,\sfive]=\sfour\mprod\sfive$, that do have the fmp, but lack the fmp w.r.t.\ product frames.

(II)
The above two examples are special cases of general results in \cite{Gabbay&Shehtman98,Shehtman04}:
If $L$ is finitely Horn axiomatisable 
(that is, finitely axiomatisable by modal formulas having universal Horn first-order correspondents), then both
$[L,\kkk]$ and $[L,\sfive]$ are determined by classes of frames that admit filtration, and so have the fmp. 

(III) 
Shehtman \cite{Shehtman14} shows that products of some modal logics of finite depth with both
$\sfive$ and $\diff$ have the fmp. He also obtains the fmp for the product logic $\diff\mprod\kkk$.

(IV)
On the negative side,
if both $L_0$ and $L_1$ are determined by transitive frames and have frames of arbitrarily large depth,
then no logic between $\commut$ and $L_0\mprod L_1$ has the fmp \cite{gkwz05a}. So for example, neither
$[\kft,\kft]$ nor $[\kft,\kfour]$ have the fmp.

(V)
Reynolds \cite{Reynolds97}  
considers the bimodal tense extension $\kft_t$ of $\kft$ as first component (that is, besides the usual `future' $\Box$, 
the language of $\kft_t$ contains a `past' modal operator as well, interpreted along the inverse of the accessibility 
relation of $\Box$).
He shows that the 3-modal product logic $\kft_t\mprod\sfive$ does not have the fmp.

\medskip
In this paper we show that commutators with a `weakly connected' component often lack the fmp.
Our results imply that (II) above cannot be generalised to component logics having weakly connected frames only: 
Even commutators without `transitive'
components such as $[\kt,\kkk]$ can lack the fmp (here $\kt$ is the logic determined by all --not necessarily transitive-- weakly connected frames). 
On the other hand, we generalise (IV) (and (V)) 
above for cases where one of the component logics have frames of modal
depth one only. In particular, we show (without using the `past' operator) that the (decidable  \cite{Reynolds97})  product logics $\kft\mprod\sfive$ and $\sft\mprod\sfive$ do not have the fmp.
Precise formulations of our results are given in Section~\ref{results}. These results give negative answers to questions in
\cite{Gabbay&Shehtman98}, and to Questions~6.43 and 6.62 in \cite{gkwz03}.

The structure of the paper is as follows. Section~\ref{defs} provides the relevant definitions and notation,
and we discuss the fmp w.r.t.\ product frames in more detail.
Our results are listed in Section~\ref{results}, and proved in Section~\ref{proofs}.
Finally, in Section~\ref{disc} we discuss the obtained results and formulate some open problems.


\section{Bimodal logics and product frames}\label{defs}

In what follows we assume that the reader is familiar with the basic notions in modal logic
and its possible world semantics (for reference, see, e.g., \cite{Blackburnetal01,cz}).
Below we
summarise some of the necessary notions and notation for the bimodal case.
Similarly to (propositional) unimodal formulas,
by a \emph{bimodal} formula we mean any formula built up from propositional variables
using the Booleans and the unary modal operators $\Bh$, $\Bv$, and $\Dh$, $\Dv$. 
Bimodal formulas are evaluated in $2$-\emph{frames}: relational structures of the form
$\F=(W,R_0,R_1)$, having two binary relations $R_0$ and $R_1$ on a non-empty set $W$.
A \emph{Kripke model based on} $\F$ is a pair $\M=(\F,\vartheta)$, where $\vartheta$ is a
function mapping propositional variables to subsets of $W$. The \emph{truth relation}
`$\M,w\models\varphi$', connecting points in models and formulas, is defined as usual by induction 
on $\varphi$. 
We say that $\varphi$ is \emph{valid in} $\F$, 
if $\M,w\models\varphi$, for every model $\M$ based on $\F$ and
for every $w\in W$. 
If every formula
in a set $\Sigma$ is valid in $\F$, then we say that $\F$ is a \emph{frame for}  $\Sigma$.
We let $\Fr \Sigma$ denote the class of all frames for $\Sigma$.

A set $L$ of bimodal formulas is called a (normal) \emph{bimodal logic} (or \emph{logic}, for short)
if it contains all propositional tautologies and the formulas
$\Box_i(p\to q)\to(\Box_i p\to\Box_i q)$, for $i<2$,
and is closed  under the rules of Substitution, Modus Ponens and 
Necessitation $\varphi/\Box_i\varphi$, for $i<2$. Given a class $\mathcal{C}$ of $2$-frames, we always
obtain a logic by taking
\[
\Log\,\mathcal{C}=\{\varphi :\varphi\mbox{ is a bimodal formula valid in every member of }\mathcal{C}\}.
\]
We say that $\Log\,\mathcal{C}$ is \emph{determined by} $\mathcal{C}$, and call such a logic
\emph{Kripke complete}. (We write just $\Log\,\F$ for $\Log\,\{\F\}$.)
%

Let $L_0$ and $L_1$ be two unimodal logics formulated using the same propositional variables and Booleans,
but having different modal operators ($\Dh$, $\Bh$ for $L_0$, and $\Dv$, $\Bv$ for $L_1$). Their \emph{fusion}
$\fusion$ is the smallest bimodal logic that contains both $L_0$ and $L_1$. The \emph{commutator} $\commut$
of $L_0$ and $L_1$ is the smallest bimodal logic that contains $\fusion$ and the formulas in \eqref{commut}.
Next, we introduce some special `two-dimensional' $2$-frames for commutators.
Given unimodal Kripke frames 
$\F_0=(W_0,R_0)$ and $\F_1=(W_1,R_1)$, their \emph{product} is defined to be
the $2$-frame
\[
\F_0\mprod\F_1= ( W_0\mprod W_1,\overline{R}_0,\overline{R}_1),
\]
where $W_0\mprod W_1$ is the Cartesian product of $W_0$ and $W_1$
and, for all $u,u'\in W_0$, $v,v'\in W_1$,
\begin{gather*}
(u,v) \overline{R}_0 (u',v')\quad \text{ iff }\quad uR_0u'\mbox{ and }v=v',\\
(u,v) \overline{R}_1 (u',v')\quad \text{ iff }\quad vR_1v' \mbox{ and }u=u'.
\end{gather*}
$2$-frames of this form will be called \emph{product frames} throughout.
For classes $\mathcal{C}_0$ and $\mathcal{C}_1$ of unimodal frames, we define
\[
\mathcal{C}_0\mprod\mathcal{C}_1=\{\F_0\mprod \F_1 : \F_i\in\mathcal{C}_i,\mbox{ for $i=0,1$}\}.
\]
Now, for $i<2$, let $L_i$ be a Kripke complete unimodal logic in the language with $\Diamond_i$ and $\Box_i$.
The \emph{product} of $L_0$ and $L_1$ is defined as the (Kripke complete) bimodal logic
\[
L_0\times L_1 =\Log\,(\Fr L_0\mprod\Fr L_1).
\]
As we briefly discussed in Section~\ref{intro}, product frames always validate the formulas in \eqref{commut},
and so $\commut\subseteq L_0\mprod L_1$ always holds. 
If both $L_0$ and $L_1$ are Horn axiomatisable, then $\commut=L_0\mprod L_1$ \cite{Gabbay&Shehtman98}. 
In general, $\commut$ can be properly contained in $L_0\mprod L_1$. In particular,
the universal (but not Horn) property of weak connectedness can result in such behaviour:
$[\kft,\kkk]$ is properly contained in the non-finitely axiomatisable $\kft\mprod\kkk$ \cite{km12},
see \cite[Thms.5.15,\,5.17]{gkwz03} and \cite{hk14} for more examples
(here $\kkk$ and $\kft$ denote the unimodal logics determined, respectively, by all frames, and by all transitive and
weakly connected frames).

It is not hard to force infinity in product frames.
The following formula \cite[Thm.5.32]{gkwz03} forces an infinite ascending $\overline{R}_0$-chain of distinct points
in product frames with a transitive first component:
\begin{equation}\label{fasc}
\Bh^+\Dv p\land \Bh^+\Bv(p\to \Dh\Bh^+\neg p)
\end{equation}
(here $\Bh^+\psi$ is  shorthand for $\psi\land\Bh\psi$).
Also, the formula 
\begin{equation}\label{fdesc}
\Dv\Dh p\land\Bv(\Dh p\to \Dh\Dh p)\land \Bv\Bh(p\to \Bh\neg p) \land \Bh\Dv p
\end{equation}
forces a rooted infinite descending $\overline{R}_0$-chain of points in product frames with a transitive and weakly connected first component
(see  \cite[Thm.6.12]{gabelaiaphd} for a similar formula).
It is not hard to see that both \eqref{fasc} and \eqref{fdesc} can be satisfied in infinite product frames, where the second component is
a \emph{one-step rooted frame} $(W,R)$ (that is, there is $r\in W$ such that $rRw$ for every $w\in W$, $w\ne r$).
As a consequence,
a wide range of bimodal logics fail to have the \emph{fmp w.r.t.\ product frames\/}. 
If every finite frame for a logic is the p-morpic image of one of its finite product frames, then the lack of fmp
follows. As is shown in  \cite{gabelaiaphd}, such examples are the logics $[\glt,L]$ and $\glt\mprod L$, 
for any $L$ having one-step rooted frames
(here $\glt$ is the logic determined by all Noetherian strict linear orders).
However, in general this is not the case for bimodal logics with frames having weakly connected components.
Take, say, the 2-frame $\F=\auf W,\leq,W\mprod W\zu$,
where $W=\{x,y\}$ and $x\leq x\leq y\leq y$. Then it is easy to see that $\F$ is a p-morphic image of 
$\auf\omega,\leq\zu\times\auf\omega,\omega\times\omega\zu$,
but $\F$ is not a p-morphic image of any finite product frame.


\section{Results}\label{results}

We denote by $\kt$  the unimodal logic determined by all
weakly connected (but not necessarily transitive) frames.
%

\begin{theorem}\label{t:wcxk}
Let $L$ be a bimodal logic such that
\begin{itemize}
\item
$[\kt,\kkk]\subseteq L$, and
\item
$\auf\omega+1,>\zu\mprod\F$ is a frame for $L$, where
$\F$ is a countably infinite one-step rooted frame.
\end{itemize}
Then $L$ does not have the finite model property.
\end{theorem}

Weak connectedness is a property of linear orders, and
$\auf\omega+1,>\zu$ is a frame for $\kft$.
Most `standard' modal logics have infinite one-step rooted frames, in particular, $\sfive$ (the logic of all equivalence frames),
and $\diff$ (the logic of all \emph{difference frames} $(W,\ne)$). So we have:

\begin{corollary}
Let $L_0$ be either $\kt$ or $\kft$, and $L_1$ be any of 
$\kkk$, $\sfive$, $\diff$. Then no logic between $\commut$ and $L_0\mprod L_1$ has the fmp.
\end{corollary}

However, $\auf\omega+1,>\zu$ is not a frame for `linear' logics whose frames are serial, reflexive and/or dense, 
such as $\Log\,(\omega,<)$, $\sft$, or the logic $\Log\,(\mathbb Q,<)=\Log\,(\mathbb R,<)$ 
of the usual orders over the rationals or the reals.
 Our next theorem deals with these kinds of logics as first components.
We say that a frame $\F = \auf W, R\zu$ \emph{contains an\/} $\auf\omega + 1, >\zu$-\emph{type chain\/},
if there are distinct points $x_n$, for $n\leq\omega$, in $W$ such that  $x_nRx_m$ iff $n>m$, 
for all $n,m\leq\omega$, $n\ne m$. 
Observe that this is less than saying that $\F$ has a subframe isomorphic to $\auf\omega + 1, >\zu$,
as for each $n$, $x_n R x_n$ might or might not hold. So $\F$ can be reflexive and/or dense, and still have this property.

\begin{theorem}\label{t:linxsfive}
Let $L$ be a bimodal logic such that
\begin{itemize}
\item
$[\kft,\kkk]\subseteq L$, and
\item
$\F_0\mprod\F_1$ is a frame for $L$, where $\F_0$ contains an $\auf\omega+1,>\zu$-type chain, and 
$\F_1$ is a countably infinite one-step rooted frame.
\end{itemize}
Then $L$ does not have the finite model property.
\end{theorem}

\begin{corollary}
Let $L_0$ be any of $\Log (\omega,<)$,  $\Log (\omega,\leq)$, $\sft$, $\Log (\mathbb Q,<)$,  and $L_1$ be any of 
$\kkk$, $\sfive$, $\diff$. Then no logic between $\commut$ and $L_0\mprod L_1$ has the fmp.
\end{corollary}

Our last theorem is about bimodal logics having less interaction than commutators.
Let 
$[L_0,L_1]^{\textit{lcom}}$ denote the smallest bimodal logic containing $L_0\oplus L_1$ and 
$\Bv\Bh p\to \Bh\Bv p$.
We denote by $\pskf$ the unimodal logic determined by all frames that are \emph{pseudo-transitive}:
\[
\forall x,y,z\in W\,\bigl(xRyRz \to (x=z\lor xRz)\bigr).
\]
Difference frames $(W,\ne)$ are examples of pseudo-transitive frames where
the accessibility relation $\ne$ is also symmetric. (Note that in 2-frames with a symmetric second relation,
(rcom) is equivalent to (conf).)

\begin{theorem}\label{t:linxdiff}
Let $L$ be a bimodal logic such that
\begin{itemize}
\item
$[\kt,\pskf]^{\textit{lcom}}\subseteq L$, and

\item
$\auf\omega + 1, >\zu\mprod\auf\omega,\ne\zu$ is a frame for $L$.
\end{itemize}
Then $L$ does not have the finite model property.
\end{theorem}

\begin{corollary}
Neither $[\kt,\pskf]^{\textit{lcom}}$ nor $[\kt,\diff]^{\textit{lcom}}$ have the fmp.
\end{corollary}


\section{Proofs}\label{proofs}

\begin{proof}[Proof of Theorem~\ref{t:wcxk}.]
For every bimodal formula $\varphi$ and every $n<\omega$, we let
\[
\Dhn\varphi\;=\;\Dh^n\varphi\land\Bh^{n+1}\neg\varphi\;=\;\overbrace{\Dh\dots\Dh}^n\varphi\land
\overbrace{\Bh\dots\Bh}^{n+1}\neg\varphi.
\]
We will use a `refinement' of the formula \eqref{fdesc}.
Let $\finfone$ be the conjunction of the following formulas:
\begin{align}
\label{init}
& \Dv\Dh (\diag\land\Bh\bot),\\
\label{hgen}
& \Bv (\Dh\diag\to\Dh\Dhe\diag),\\
\label{vgen}
& \Bh\bigl(\Dv\Dhe\diag\to\Dv(\diag\land\Bh\neg\diag\land\Bh\Bh\neg\diag)\bigr).
\end{align}

\begin{lemma}\label{l:infone}
Let $\F=\auf W,\Rh,\Rv\zu$ be any 2-frame such that $\Rh$ is weakly connected, and 
$\Rh$, $\Rv$ are confluent and commute. If $\finfone$ is satisfiable in $\F$, then $\F$ is infinite.
\end{lemma}

\begin{proof}
We will only use the following consequence of weak connectedness:

\medskip
\noindent
\wcon\qquad
\parbox[t]{12cm}
{$\forall x,y,z\ \Bigl(x\Rh y\;\land\; x\Rh z\; \to\; \bigl(y\Rh z\;\lor\; z\Rh y\;\lor\;\forall w\, (y\Rh w\;\leftrightarrow\; z\Rh w)\bigr)\Bigr).$}

\medskip
\noindent
Suppose that $\M,r\models\finfone$ for some model $\M$ based on $\F$. First,
we define inductively three sequences $u_n$, $v_n$, $x_n$, for $n<\omega$,
of points in $\F$ such that, for every $n<\omega$,
\begin{itemize}
\item[(a)]
$v_n\Rh u_n$,

\item[(b)]
$r\Rh x_n \Rv v_n$, and if $n>0$ then $x_{n-1}\Rv u_n$,

\item[(c)]
$\M,u_n\models\diag\land\Bh\neg\diag\land\Bh\Bh\neg\diag$,

\item[(d)]
$\M,v_n\models\Dhe\diag$.
\end{itemize}
If $n=0$, then by \eqref{init} there are $y_0$, $u_0$ such that
$r\Rv y_0\Rh u_0$ and
\begin{equation}\label{endpoint} 
\M,u_0\models \diag\land \Bh\bot,
\end{equation}
and so (c) holds.
By \eqref{hgen}, there is $v_0$ such that $y_0\Rh v_0$ and $\M,v_0\models\Dhe\diag$, and so
$v_0\Rh u_0$ follows by \wcon\ and \eqref{endpoint}. By \rcom, we have $x_0$ with $r\Rh x_0\Rv v_0$.

Now suppose that, for some $n <\omega$, $u_i$, $v_i$, $x_i$ with (a)--(d)
have already been defined for all $i\leq n$. 
By (b) and (d) of the IH, $r\Rh x_n$ and $\M,x_n\models\Dv\Dhe\diag$. So by \eqref{vgen}, there is $u_{n+1}$ such 
that $x_n\Rv u_{n+1}$ and 
\begin{equation}\label{nodiag}
\M,u_{n+1}\models\diag\land\Bh\neg\diag\land\Bh\Bh\neg\diag.
\end{equation}
By \lcom, there is $y_{n+1}$ with $r\Rv y_{n+1}\Rh u_{n+1}$. By \eqref{hgen}, 
there is $v_{n+1}$ such that $y_{n+1}\Rh v_{n+1}$ and $\M,v_{n+1}\models\Dhe\diag$, and so
$v_{n+1}\Rh u_{n+1}$ follows by \wcon\ and \eqref{nodiag}.
By \rcom, we have $x_{n+1}$ with $r\Rh x_{n+1}\Rv v_{n+1}$.

Next, we show that all the $u_n$ are different, and so $\F$ is infinite.
We show by induction on $n$ that, for all $n <\omega$,
\begin{equation}\label{diff}
\M,u_n\models\Dhn\top.
\end{equation}
For $n=0$,  \eqref{diff} holds by \eqref{endpoint}. Suppose inductively that \eqref{diff} holds for some
$n<\omega$.  We have $v_n\Rh u_n$, by (a) above. We claim that
\begin{equation}\label{vgood}
\forall u\ (v_n\Rh u\ \to\ \M,u\models\Bh^{n+1}\bot).
\end{equation}
Indeed, suppose that $v_n\Rh u$.
By \wcon, we have either $u\Rh u_n$, or $u_n\Rh u$, or  
$\forall w\, (u_n\Rh w\;\leftrightarrow\; u\Rh w)$. 
As $\M,u_n\models\diag$ by (c), and $\M,v_n\models\Bh\Bh\neg\diag$ by (d),
we cannot have $u\Rh u_n$.
As we have $\M,u_n\models\Bh^{n+1}\bot$ by the IH,
in the other two cases $\M,u\models\Bh^{n+1}\bot$ follows, proving \eqref{vgood}. 
As $\M,u_n\models\Dh^n\top$ by the IH, we obtain
\begin{equation}\label{vok}
 \M,v_n\models\Dhnn\top
 \end{equation}
 by \eqref{vgood} and (a).
By (b), we have $r\Rh x_n\Rv v_n$ and $x_n\Rv u_{n+1}$. 
So $\M,x_{n}\models\Dh^{n+1}\top$ follows by \rcom\ and \eqref{vok}.
Also, by \conf\ and \eqref{vok}, we have $\M,x_n\models\Bh^{n+2}\bot$.
Now we have $\M,u_{n+1}\models\Dh^{n+1}\top$ by \conf, and 
$\M,u_{n+1}\models\Bh^{n+2}\bot$ by \rcom.
Therefore, $\M,u_{n+1}\models\Dhnn\top$, as required.
\end{proof}

\begin{lemma}\label{l:satone}
Let $\F$ be a countably infinite one-step rooted frame. 
Then $\finfone$ is satisfiable in $\auf\omega+1,>\zu\mprod\F$.
\end{lemma}

\begin{proof}
Suppose $\F=\auf W,R\zu$, and let $r,y_0,y_1,\dots$ be an arbitrary enumeration of $W$. Define a model $\M$ over 
$\auf\omega+1,>\zu\mprod\F$ by taking
\[
\M,\auf n,y\zu\models\diag  \qquad\mbox{iff}\qquad
n<\omega,\ y=y_n.
\]
Then it is straightforward to check that $\M,\auf \omega,r\zu\models\finfone$.
\end{proof}

Now Theorem~\ref{t:wcxk} follows from Lemmas~\ref{l:infone} and \ref{l:satone}.
\end{proof}


\medskip
\begin{proof}[Proof of Theorem~\ref{t:linxsfive}.]
We will use a variant of the formula $\finfone$ used in the previous proof. The problem is that in reflexive and/or dense 
frames, a formula of the form $\Dhe p$ is clearly not satisfiable.
In order to fix this, we use a version of the `tick trick', introduced in \cite{Spaan93,gkwz05a}.
We fix a propositional variable $t$, and define a new modal operator by setting, for every formula $\psi$,
\begin{align*}
& \DhM\psi=\bigl [t\to\Dh\bigl(\neg t\land (\psi\lor\Dh\psi)\bigr)\bigr]\land
\bigl [\neg t\to\Dh\bigl(t\land (\psi\lor\Dh\psi)\bigr)\bigr],\mbox{ and }\\
& \BhM\phi=\neg\DhM\neg\psi.
\end{align*}
Now let $\M$ be a model based on some 2-frame $\F=\auf W,\Rh,\Rv\zu$. We define a new binary relation
$\RhM$ on $W$ by taking, for all $x,y\in W$,
\[
x\RhM y\quad\mbox{iff}\quad
\exists z\in W\ \bigl(x\Rh z\mbox{ and }(\M,x\models t\ \leftrightarrow\ \M,z\models\neg t)\mbox{ and }
(z=y\mbox{ or }z\Rh y)\bigr).
\]
We will write $x\neg\RhM y$, whenever $x\RhM y$ does not hold.
It is straightforward to check the following:

\begin{claim}\label{c:trans}
If $\Rh$ is transitive, then $\RhM$ is transitive as well,
$\RhM\subseteq \Rh$,
$\Rh\circ\RhM\subseteq \RhM$, and
$\RhM\circ \Rh\subseteq \RhM$.
\end{claim}

 Also, $\DhM$ behaves like
a modal diamond w.r.t. $\RhM$, that is, for all $x\in W$,
\[
\M,x\models\DhM\psi\quad\mbox{iff}\quad
\exists y\in W\ \bigl(x\RhM y\mbox{ and }\M,y\models\psi\bigr).
\]
However, $\RhM$ is not necessarily weakly connected whenever $\Rh$ is weakly connected, but if $\Rh$ is also transitive, then it does have 

\medskip
\noindent
\wconM\quad
\parbox[t]{12cm}
{$\forall x,y,z\ \Bigl(x\RhM y\;\land\; x\RhM z\; \to\; \bigl(y\RhM z\;\lor\; z\RhM y\;\lor\;\forall w\, (y\RhM w\;\leftrightarrow\; z\RhM w)\bigr)\Bigr).$}

\begin{claim}\label{c:wcon}
If $\Rh$ is transitive and weakly connected, then \wconM\ holds in $\M$.
\end{claim}

\begin{proof}
Suppose that $x\RhM y$ and $x\RhM z$. By Claim~\ref{c:trans} and weak connectedness of $\Rh$, we have
that either $y=z$, or $y\Rh z$, or $z\Rh y$.
If $y=z$ then $\forall w\, (y\RhM w\;\leftrightarrow\; z\RhM w)$ clearly holds.
Next, suppose $y\Rh z$ and $y\neg\RhM z$. We claim that $\forall w\, (y\RhM w\;\leftrightarrow\; z\RhM w)$ follows.
Indeed, suppose first that $z\RhM w$ for some $w$. Then we have $y\RhM w$ by Claim~\ref{c:trans}.
Now suppose $y\RhM w$ for some $w$, and $\M,y\models t$. (The case when $\M,y\models\neg t$ is similar.) 
As $y\Rh z$ and $y\neg\RhM z$, we also have $\M,z\models t$. Further, there is $u$ such that $\M,u\models\neg t$,
$y\Rh u$ and either $u=w$ or $u\Rh w$. As $\Rh$ is weakly connected, either $u=z$, or $u\Rh z$, or $z\Rh u$.
As $y\Rh z$ and $y\neg\RhM z$, we cannot have $u=z$ or $u\Rh z$, and so $z\Rh u$ follows, implying $z\RhM w$ as required. The case when $z\Rh y$ and $z\neg\RhM y$ is similar.
\end{proof}

In case $\Rh$ and $\Rv$ interact in certain ways, we would like to force similar interactions between 
$\RhM$ and $\Rv$. To this end, suppose that $\M,r\models\eqref{tick}$, where
\begin{equation}\label{tick}
(t\lor\Dv t\to t\land\Bv t)\land \Bh(t\lor\Dv t\to t\land\Bv t),
\end{equation}
and consider the following properties:

\medskip
\noindent
\lcomM\quad
$\forall y,z\ (r\RhM y\Rv z\;\to\;\exists u\ r\Rv u\RhM z),$

\medskip
\noindent
\rcomM\quad
$\forall x,y,z\ \bigl((x=r\;\lor\;r\Rh x)\;\land\;x\Rv y\RhM z\;\to\;\exists u\ x\RhM u\Rv z\bigr),$

\medskip
\noindent
\confM\quad\,
$\forall x,y,z\ \bigl(r\Rh x\RhM z\;\land\;x\Rv y\;\to\;\exists u\; (y\RhM u\;\land\; z\Rv u)\bigr).$

\begin{claim}\label{c:comm}
Suppose that $\Rh$ is transitive and $\M,r\models\eqref{tick}$.
\begin{itemize}
\item[{\rm (i)}]
If \lcom\ holds in $\F$, then \lcomM\ holds in $\M$.
\item[{\rm (ii)}]
If \rcom\ holds in $\F$, then \rcomM\ holds in $\M$.
\item[{\rm (iii)}]
If \conf\ holds in $\F$, then \confM\ holds in $\M$.
\end{itemize}
\end{claim}

\begin{proof}
We show (ii) (the proofs of the other two items are similar and left to the reader).
Suppose that $x=r$ or $r\Rh x$, $x\Rv y\RhM z$, and $\M,x\models t$.
Then by \eqref{tick}, we have $\M,y\models t$. As $y\RhM z$, there is $v$ such that 
$\M,v\models \neg t$, $y\Rh v$, and $v=z$ or $v\Rh z$. By \rcom, there is $w$ with $x\Rh w\Rv v$,
and so $\M,w\models \neg t$ by the transitivity of $\Rh$ and \eqref{tick}.
If $v=z$, then $x\RhM w\Rv z$, as required. If $v\Rh z$ then, again by \rcom, there is $u$ with
$w\Rh u\Rv z$. Therefore, $x\RhM u\Rv z$, as required. The case when $\M,x\models\neg t$ is similar.
\end{proof}

Let $\finfoner$ be the conjunction of \eqref{tick} and the formulas obtained from \eqref{init}--\eqref{vgen}
by replacing each $\Dh$ with $\DhM$, and each $\Bh$ with $\BhM$. Now, because of Claims~\ref{c:wcon} and \ref{c:comm}, the following lemma is proved analogously to Lemma~\ref{l:infone}, with replacing $\Rh$ by $\RhM$
everywhere in its proof:

\begin{lemma}\label{l:infoner}
Let $\F=\auf W,\Rh,\Rv\zu$ be any 2-frame such that $\Rh$ is transitive and weakly connected, and 
$\Rh$, $\Rv$ are confluent and commute. If $\finfoner$ is satisfiable in $\F$, then $\F$ is infinite.
\end{lemma}

\begin{lemma}\label{l:satoner}
Let $\F_0$ be a frame for $\kft$ that contains an $\auf\omega+1,>\zu$-type chain, and let $\F_1$ be a countably
infinite one-step rooted frame. Then $\finfoner$ is satisfiable in $\F_0\mprod\F_1$.
\end{lemma}

\begin{proof}
Suppose $\F_i=\auf W_i,R_i\zu$ for $i=0,1$. Let $x_n$, for $n\leq\omega$, be distinct points in $W_0$ such that
for all $n,m\leq\omega$, $n\ne m$, we have $x_n \Rh x_m$ iff $n>m$. For every $n<\omega$, we let
\[
[x_{n+1},x_n)=\bigl(\{x\in W_0 : x_{n+1} \Rh x \Rh x_n\}\cup\{x_{n+1}\}\bigr)-
\{x : x=x_{n}\mbox{ or }x_n\Rh x\}.
\]
Let $r,y_0,y_1,\dots$ be an arbitrary enumeration of $W_1$. Define a model $\M$ over $\F_0\mprod\F_1$ by taking
\begin{align*}
\M,\auf x,y\zu\models t & \quad\mbox{iff}\quad
x\in [x_{n+1},x_n),\ n<\omega,\ \mbox{$n$ is odd},\ y\in W_1,\\
\M,\auf x,y\zu\models\diag & \quad\mbox{iff}\quad
x\in [x_{n+1},x_n),\ y=y_n,\ n<\omega.
\end{align*}
Then it is easy to check that $\M,\auf x_\omega,r\zu\models\finfoner$.
\end{proof}

Now Theorem~\ref{t:linxsfive} follows from Lemmas~\ref{l:infoner} and \ref{l:satoner}.
\end{proof}


\medskip
\begin{proof}[Proof of Theorem~\ref{t:linxdiff}.]
Let $\finftwo$ be the conjunction of the following formulas:
\begin{align}
\label{init2}
& \Dh(\diag\land\neg\next\land\Bh\neg\next\land\Bv\neg\next),\\
\label{hgen2}
& \Bv^+\Dh(\next\land\Bv\neg\next),\\
\label{vgen2}
& \Bv^+\Bh\bigl(\next\to\Dv(\diag\land\neg\next\land\Bh\neg\next\land\Dv\next)\bigr),\\
\label{puniq}
& \Bv^+\Bh\Bh(\diag\to\Bh\neg\diag),
\end{align}
where $\Bv^+\psi=\psi\land\Bv\psi$, for any formula $\psi$.

\begin{lemma}\label{l:inftwo}
Let $\F=\auf W,\Rh,\Rv\zu$ be any 2-frame such that $\Rh$ is weakly connected, $R_1$ is
pseudo-transitive, and $\Rh$, $\Rv$ left-commute. If $\finftwo$ is satisfiable in $\F$, then $\F$ is infinite.
\end{lemma}

\begin{proof}
Suppose that $\M,r\models\finftwo$ for some model $\M$ based on $\F$. First,
we define inductively three sequences $y_n$, $u_n$, $v_n$, for $n<\omega$,
of points in $\F$ such that, for every $n<\omega$,
\begin{itemize}
\item[(e)]
($y_n=r$ or $r\Rv y_n$), and $y_n \Rh v_n\Rh u_n$, 

\item[(f)]
if $n>0$, then $v_{n-1}\Rv u_n$ and $u_n\Rv v_{n-1}$,

\item[(g)]
$\M,u_n\models\diag$,

\item[(h)]
$\M,v_n\models\next\land\Bv\neg\next$.
\end{itemize}
If $n=0$, then let $y_0=r$. By \eqref{init2}, there is $u_0$ such that
$y_0\Rh u_0$ and
\begin{equation}\label{endpoint2} 
\M,u_0\models \diag\land\neg\next\land\Bh\neg\next\land\Bv\neg\next.
\end{equation}
By \eqref{hgen2}, there is $v_0$ such that $y_0\Rh v_0$ and $\M,v_0\models \next\land\Bv\neg\next$.
Thus $v_0\Rh u_0$ follows by the weak connectedness of $\Rh$ and \eqref{endpoint2}.

Now suppose that, for some $n <\omega$, $y_i$, $u_i$, $v_i$ with (e)--(h)
have already been defined for all $i\leq n$. 
By (e) and (h) of the IH,
either $y_n=r$ or $r\Rv y_n$, $y_n\Rh v_n$ and $\M,v_n\models\next\land\Bv\neg\next$. 
Also, by \eqref{vgen2} there is $u_{n+1}$ such that $v_n\Rv u_{n+1}$ and 
\begin{equation}\label{uok}
\M,u_{n+1}\models\diag\land\neg\next\land\Bh\neg\next\land\Dv\next,
\end{equation}
 and so $u_{n+1}\Rv v_n$ follows by the pseudo-transitivity of $\Rv$.
By \lcom, there is $y_{n+1}$ such that $y_n\Rv y_{n+1}\Rh u_{n+1}$. 
By the pseudo-transitivity of $\Rv$ and (e) of the IH, we have $y_{n+1}=r$ or $r\Rv y_{n+1}$.
Now by \eqref{hgen2}, there is $v_{n+1}$ such that $y_{n+1}\Rh v_{n+1}$ and 
$\M,v_{n+1}\models\next\land\Bv\neg\next$.
As $\M,u_{n+1}\models\neg\next\land\Bh\neg\next$ by \eqref{uok},
$v_{n+1}\Rh u_{n+1}$ follows by the weak connectedness of $\Rh$.

Next, we show that all the $u_n$ are different, and so $\F$ is infinite.
We show by induction on $n$ that, for all $n <\omega$,
\begin{equation}\label{diff2}
\M,u_n\models\chi_n\land\bigwedge_{i<n}\neg\chi_i,
\end{equation}
where $\chi_0=\Bv\neg\next$, and for $n>0$,
\[
\chi_n=\Dv\bigl(\next\land\Dh(\diag\land\chi_{n-1})\bigr).
\]
For $n=0$,  \eqref{diff2} holds by \eqref{endpoint2}.
 Suppose inductively that \eqref{diff2} holds for some
$n<\omega$. 
On the one hand, as $\M,u_n\models\chi_n$ by the IH, and $u_{n+1}\Rv v_n\Rh u_n$ by (e) and (f),
we have $\M,u_{n+1}\models\chi_{n+1}$
by (h) and (g).
On the other hand, as $v_n\Rv u_{n+1}$ by (f), and $\M,v_n\models\Bv\neg\next$ by (h),
by the pseudo-transitivity of $\Rv$ we have
\begin{equation}\label{qok}
\forall w\ (u_{n+1}\Rv w\ \land\ \M,w\models\next\ \to\ w=v_{n}).
\end{equation}
Also, by (e), (g), \eqref{puniq}, and the weak connectedness of $\Rh$, we have
\begin{equation}\label{pok}
\forall w\ (v_{n}\Rh w\ \land\ \M,w\models\diag\ \to\ w=u_{n}).
\end{equation}
As $\M,u_n\models\bigwedge_{i<n}\neg\chi_i$ by the IH, we obtain that 
$\M,u_{n+1}\models\bigwedge_{i<n+1}\neg\chi_i$
by \eqref{qok} and \eqref{pok}.
\end{proof}

\begin{lemma}\label{l:sattwo}
$\finftwo$ is satisfiable in 
$\auf\omega+1,>\zu\mprod\auf \omega,\ne\zu$.
\end{lemma}

\begin{proof}
%
We define a model $\M$ over $\auf\omega+1,>\zu\mprod\auf \omega,\ne\zu$ by taking
\begin{align*}
\M,\auf m,n\zu\models\diag & \qquad\mbox{iff}\qquad
m=n,\ n<\omega,\\
\M,\auf m,n\zu\models\next & \qquad\mbox{iff}\qquad
m=n+1,\ n<\omega.
\end{align*}
Then it is easy to check that $\M,\auf \omega,0\zu\models\finftwo$.
\end{proof}

Now Theorem~\ref{t:linxdiff} follows from Lemmas~\ref{l:inftwo} and \ref{l:sattwo}.
\end{proof}


\section{Discussion and open problems}\label{disc}

We showed that commutators and products with a `weakly connected component' (that is,
a component logic having only weakly connected frames) often lack the fmp.
We conclude the paper with a discussion of related results and open problems.

\smallskip
(I) First, we discuss the
\emph{decision problem} of the logics under the scope of our results:
\begin{itemize}
\item
If $L_0$ is any of
$\kft$, $\sft$, $\Log\auf\mathbb Q,<\zu$ and $L_1$ is either $\sfive$ or $\kkk$, then $L_0\mprod L_1$ is
decidable \cite{Reynolds97,Wolter99,gkwz03}. The known proofs build product models or quasimodels (two-dimensional
structures of types) from finitely many repeating small pieces (mosaics). Can mosaic-style proofs 
be used to show that the corresponding commutators are decidable?

\item
The decidability of $\Log\bigl(\{(\omega,<)\}\mprod\Fr\sfive\bigr)$ can also be shown by a mosaic-style proof
\cite{gkwz03}.
However, in \cite[Thm.6.29]{gkwz03} it is wrongly stated that this logic is the same as $\Log\,(\omega,<)\mprod\sfive$.
Unlike richer temporal languages, the unimodal language having a single $\Diamond$ (and its $\Box$)
is not capable to capture discreteness of a linear order (though, it can forbid the existence of infinite ascending
chains between any two points). In particular, $\Log\,(\omega,<)$ does have frames containing
$\auf\omega+1,>\zu$-type chains. Therefore, the formula $\finfoner$ used in the proof of Theorem~\ref{t:linxsfive}
is $\Log\,(\omega,<)\mprod\sfive$-satisfiable by Lemma~\ref{l:satoner}. However, $\finfoner$ is not
$\Log\bigl(\{(\omega,<)\}\mprod\Fr\sfive\bigr)$-satisfiable, as  by the proof of Lemma~\ref{l:infoner}, any 2-frame
with a linear first component satisfying $\finfoner$ must contain an $\auf\omega+1,>\zu$-type chain.
So  in fact  it is not known whether any of $\Log\,(\omega,<)\mprod\sfive$ or $[\Log\,(\omega,<),\sfive]$ is decidable.
Do they have the fmp?
Also, are $\glt\mprod\sfive$ and $[\glt,\sfive]$ decidable? The similar questions for $\kkk$ in place of $\sfive$
are also open.

\item
If $L$ is any bimodal logic such that $[\kft,\diff]\subseteq L$ and the product of an infinite linear order and and infinite
difference frame is a frame for $L$, then $L$ is undecidable \cite{hk13}. Can this result be generalised to the
logics in Theorem~\ref{t:linxdiff}? In particular, is $[\kft,\diff]^{\textit{lcom}}$ decidable?

\item
It is shown in \cite{Marx&Reynolds99,Reynolds&Z01} that
if both $L_0$ and $L_1$ are determined by linear frames and have frames of arbitrary size, then
$L_0\mprod L_1$ is undecidable.
These results are generalised in \cite{gkwz05a}:
If both $L_0$ and $L_1$ are determined by transitive frames and have frames of arbitrarily large depth,
then all logics between $\commut$ and $L_0\mprod L_1$ are undecidable.
\end{itemize}

\smallskip
(II)
As the formulas in \eqref{commut} of Section~\ref{intro} are Sahlqvist-formulas, the commutator of two
Sahlqvist-axiomatisable logics is always \emph{Kripke complete}. In general, this is not the case.
Several of the commutators under the scope of the undecidability results in \cite{gkwz05a} are in fact
$\Pi_1^1$-hard, even when both component logics are finitely axiomatisable (e.g.,
$[\glt,\kfour]$ and $[\Log\auf\omega,<\zu,\kfour]$ are such).
As the commutator of two finitely axiomatisable logics is clearly recursively enumerable, 
the Kripke incompleteness of these commutators follow.
It is not known, however, whether any of the commutators
$[\glt,\sfive]$, $[\glt,\kkk]$, $[\Log\auf\omega,<\zu,\sfive]$, $[\Log\auf\omega,<\zu,\kkk]$ is
Kripke complete.

\smallskip
(III)
Apart from Theorem~\ref{t:linxdiff} above, not much is known about the fmp of bimodal logics with a weakly connected
component that are \emph{properly between fusions and commutators}. Say, does the logic of two commuting
(but not necessarily confluent) $\kft$-operators have the fmp?


\paragraph{Acknowledgments.}
I am grateful to the anonymous referee for his thorough (and quick) reading of the manuscript.

\bibliographystyle{plain} 


\end{document}